\documentclass[11pt, letter]{article}
\usepackage[margin=1in]{geometry}
\usepackage[numbers]{natbib}

\usepackage{authblk}
\usepackage{comment}
\usepackage[T1]{fontenc}
\usepackage[noend]{algorithmic}
\usepackage{algorithm}
\usepackage{amssymb}
\usepackage{pifont}
\newcommand{\cmark}{\ding{51}}%
\newcommand{\xmark}{\ding{55}}%

\usepackage{amsmath}

\usepackage{amsthm}
\usepackage{enumerate}
\usepackage{hyperref}
\usepackage[all]{hypcap}
\usepackage[referable]{threeparttablex}
\usepackage{footnote}
\usepackage{bigstrut}
\usepackage{multicol}
\usepackage{todonotes}
\usepackage{cleveref}
\usepackage[symbol]{footmisc}
\usepackage{thm-restate}
\usepackage[tableposition=bottom]{caption}
\usepackage{booktabs}

\makeatletter
\g@addto@macro{\maketitle}{\@thanks}
\makeatother

\theoremstyle{plain}
\newtheorem{thm}{Theorem}[section]

\newtheorem{lem}[thm]{Lemma}
\newtheorem{Def}[thm]{Definition}

\begin{document}

\author[1]{Moab Arar}
\author[1]{Shiri Chechik\thanks{Work done in part while the authors were at the Simons Institute for the Theory of Computing.}}
\author[1]{Sarel Cohen}
\author[2]{Cliff Stein$^{\footnotesize*}$\thanks{Supported in part by NSF grant CCF-1421161.}}
\author[3]{David Wajc$^{\footnotesize*}$\thanks{Supported in part by NSF grants CCF-1527110, CCF-1618280 and NSF CAREER award CCF-1750808.}}
\affil[1]{Tel Aviv University}
\affil[2]{Columbia University}
\affil[3]{Carnegie Mellon University} 

\title{Dynamic Matching: Reducing Integral Algorithms to Approximately-Maximal Fractional Algorithms}

\maketitle

\begin{abstract}
	We present a simple randomized reduction from fully-dynamic integral matching algorithms to fully-dynamic ``approximately-maximal'' fractional matching algorithms. Applying this reduction to the recent fractional matching algorithm of Bhattacharya, Henzinger, and Nanongkai (SODA 2017), we obtain a novel result for the integral problem. Specifically, our main result is a randomized fully-dynamic $(2+\epsilon)$-approximate integral matching algorithm with small polylog worst-case update time. For the $(2+\epsilon)$-approximation regime only a \emph{fractional} fully-dynamic $(2+\epsilon)$-matching algorithm with worst-case polylog update time was previously known, due to Bhattacharya et al.~(SODA 2017). 
	Our algorithm is the first algorithm that maintains approximate matchings with worst-case update time better than polynomial, for any constant approximation ratio. As a consequence, we also obtain the first constant-approximate worst-case polylogarithmic update time maximum weight matching algorithm.
\end{abstract}

\renewcommand*{\thefootnote}{\arabic{footnote}}

\section{Introduction}

The maximum matching problem is one of the most widely-studied problems in computer science and operations research, with a long history and theory \cite{ahuja1993network,lovasz2009matching}.
On $n$-vertex and $m$-edge graphs, the state-of-the art maximum matching algorithms require
$O(m\sqrt n)$ and $O(n^\omega)$ time \cite{micali1980v,mucha2004maximum} (here $\omega<2.37$ is the matrix multiplication exponent \cite{williams2012multiplying}). For bipartite graphs, simpler algorithms with the same asymptotic running times are known \cite{hopcroft1971n5,mucha2004maximum}, as well as a faster, $O(m^{10/7}\cdot \text{poly} (\log n))$-time algorithm, due to the recent breakthrough of \citet{madry2013navigating} for the maximum flow problem. For approximate matchings, it is long- and well-known that a matching admitting no augmenting paths of length $O(1/\epsilon)$ forms a $(1+\epsilon)$-approximate maximum matching (see \cite{hopcroft1971n5}). The linear-time ``blocking flow'' subroutines of \cite{hopcroft1971n5,micali1980v} therefore result in an $O(m/\epsilon)$-time $(1+\epsilon)$-approximate maximum matching.

The maximum weight matching (MWM) problem has also garnered much interest over the years. For general weights, the seminal work of \citet{edmonds1972theoretical} shows how to reduce the problem on bipartite graphs to the solution of $n$ non-negative single-source shortest path instances. Relying on Fibonacci Heaps of \citet{fredman1987fibonacci}, this approach yields the current fastest strongly-polynomial running time for the problem, $O(n(m+n\log n))$.  \citet{gabow1990data} later showed how to obtain the same running time for general graphs.
For integer weights $w:E\rightarrow \{0,1,2,\dots,N\}$, algorithms nearly matching the state-of-of-the-art for the unweighted problem, with either logarithmic or linear dependence on $N$, are known.\footnote{Indeed, a black-box reduction of \citet{pettie2012simple} from maximum weight matching to the maximum matching problem shows that a linear dependence in $N$ is the largest possible gap between these two problems.} These include an $O(m\sqrt n \log (nN))$-time algorithm \cite{gabow1989faster}, an $O(Nn^\omega)$-time algorithm \cite{sankowski2009maximum} and a recent $O(m^{10/7}\cdot \text{poly}(\log n) \cdot \log N)$-time algorithm for bipartite graphs \cite{cohen2017negative}. For approximation algorithms, an algorithm nearly matching the unweighted problem's guarantees is known, yielding a $(1+\epsilon)$-approximate maximum weight matching in $O((m/\epsilon) \log (1/\epsilon))$ time, \cite{duan2014linear}.

All of the above results pertain to the static problem; i.e., where the input is given and we only need to compute a maximum matching on this given input. However, in many applications the graphs considered are inherently \emph{dynamic}, with edges removed or added over time. One could of course address such changes by recomputing a solution from scratch, but this could be wasteful and time-consuming, and such applications may require immediately updating the solution given, as having users wait on a solution to be recomputed may likely be unsatisfactory. Consider for example point to point shortest path computation, a problem routinely solved by navigation systems: for such an application, the temporary closure of some road due to construction should not result in unresponsive GPS applications, busy re-computing the relevant data structures (see e.g.,\cite{AuItMaNa91,King99,HeKi01,DeIt04,DeIt06,Thorup04-Fully,Thorup05,BaswanaHS07,GrWi12,HenzingerKN13,AbrahamC13,HeKrNa14,HenzingerKNFOCS14,HenzingerKNSTOC14,HenzingerKN15,AbChDeGoWe16,AbChKr17}). Therefore, for such applications we want to update our solution quickly for \emph{every update}, using fast \emph{worst-case} (rather than amortized) update time.

Returning to the maximum matching problem, we note that a maximum matching can be trivially updated in $O(m)$ time. \citet{sankowski2007faster} showed how to maintain the \emph{value} of the maximum matching in $O(n^{1.495})$ update time.\footnote{We emphasize that this algorithm does not maintain an actual matching, but only the optimal value, and it seems unlikely to obtain such update times for maintaining a matching of this value.} On the other hand,
\citet{abboud2014popular} and \citet{kopelowitz2016higher} presented lower bounds based on long-standing conjectures, showing that even maintaining the maximum matching value likely requires $\Omega(m^c)$ update time for some constant $c\geq \frac{1}{3}$.

Given these hardness results for exact solutions, one is naturally inclined to consider fast \emph{approximate} solutions. Trivially updating a maximal matching (and therefore a $2$-approximate maximum matching) can be done using $O(n)$ worst-case update time. The goal is to obtain sublinear update times -- ideally  polylogarithmic (or even constant) -- with as low an approximation ratio as possible.

The first non-trivial result for fully-dynamic maximum matching is due to \citet{ivkovic1993fully}, who presented a maximal matching algorithm with $O((m+n)^{1/\sqrt2})$ amortized update time. Note that this bound is sublinear only for sufficiently sparse graphs. The problem of approximate maximum matchings remained largely overlooked until 2010, when \citet{onak2010maintaining} presented a fully-dynamic constant-approximate $O(\log^2n)$ (amortized) update time algorithm. Additional results followed in quick succession.

\citet{baswana2011fully} showed how to maintain a maximal matching in $O(\log n)$ expected update time, and $O(\log^2n)$ update time w.h.p. This was recently improved by \citet{solomon2016fully} who presented  a maximal matching algorithm using $O(1)$ update time w.h.p. For deterministic algorithms, \citet{neiman2013simple} showed how to maintain $3/2$-approximate matchings  deterministically in $O(\sqrt{m})$ update time, a result later improved by \citet{gupta2013fully} to obtain $(1+\epsilon)$-approximate matchings in $O(\sqrt{m}/\epsilon^2)$. This result was in turn refined by \citet{peleg2016dynamic}, who obtained the same approximation ratio and update time as \cite{gupta2013fully} with $\sqrt{m}$ replaced by the maximum arboricity of the graph $\alpha$ (which is always at most $\alpha=O(\sqrt{m})$). \citet{bernstein2015fully,bernstein2016faster,bhattacharya2015deterministic} presented faster polynomial update time algorithms (with higher approximation ratios), and \citet{bhattacharya2016new} presented a $(2+\epsilon)$-approximate algorithm with $\text{poly}(\log n,\epsilon^{-1})$ update time.
See
\Cref{table:Results} for an in-depth tabular exposition of previous work and our results.\footnote{For the sake of simplicity we only list bounds here given in terms of $n$ and $m$. In particular, we do not state the results for arboricity-bounded graphs, which in the worst case (when the arboricity of a graph is $\alpha=\Theta(\sqrt m)$) are all outperformed by algorithms in this table, with the aforementioned algorithm of \citet{peleg2016dynamic} being the lone exception to this rule.} In  \S\ref{sec:mwm} we discuss our results for MWM, also widely studied in the dynamic setting (see, e.g. \cite{baswana2011fully,gupta2013fully,solomon2016fully,stubbs2017metatheorems}).

Note that in the previous paragraph we did not state whether the update times of the discussed algorithms were worst case or amortized. We now address this point.
As evidenced by \Cref{table:Results}, previous fully-dynamic matching algorithms can be broadly divided into two classes according to their update times: polynomial update time algorithms and polylogarithmic \emph{amortized} update time algorithms. The only related polylogarithmic worst-case update time algorithms known to date were \emph{fractional} matching algorithms, due to \citet{bhattacharya2017fully}. 
We bridge this gap by presenting the first fully-dynamic integral matching (and weighted matching) algorithm with polylogarithmic worst-case update times and constant approximation ratio. In particular, our approach yields a $(2+\epsilon)$-approximate algorithm, within the $O_\epsilon(\log^3 n)$ time bound of \cite{bhattacharya2017fully}, but for \emph{integral} matching.\footnote{\label{charikar-slomon-footnote}Independently of our work, and using a different approach, Charikar and Solomon \cite{charikar2017fully} obtained a $(2+\epsilon)$-approximate dynamic matching algorithm with $O_\epsilon(\log^7 n)$ worst-case update time. For fixed $\epsilon$ their algorithm is slower than ours, and is arguably more complicated than our approach.}

\vspace{-0.35cm}
	\begin{center}	
	\begin{table}[h]
	{
		\small
			\begin{center}
			\centering
			
			\begin{tabular}{ | c | c | c | c | c | c |}
				\hline				
				Approx. & Update Time & det. & w.c. & notes  & reference \bigstrut\\  
				\specialrule{.125em}{.0625em}{.0625em} 	
				
				$O(1)$  & $O(\log^2n)$ & \xmark & \xmark &   & Onak and Rubinfeld (STOC '10) \cite{onak2010maintaining}\bigstrut \\ 		\hline			
				
				$4+\epsilon$ & $O(m^{1/3}/\epsilon^2)$ & \cmark & \cmark & 
				& Bhattacharya et al.~(SODA '15) \cite{bhattacharya2015deterministic}\bigstrut \\ 		\hline
				
				$3+\epsilon$ &
				$O(\sqrt{n}/\epsilon)$ & \cmark & \xmark & 
				& Bhattacharya et al.~(SODA '15) \cite{bhattacharya2015deterministic}\bigstrut \\ 					
				\specialrule{.125em}{.0625em}{.0625em}

				$2+\epsilon$ & $\text{poly}(\log n,1/\epsilon)$  & \cmark & \xmark & & Bhattacharya et al.~(STOC '16) \cite{bhattacharya2016new} \\ \hline
				
				$2+\epsilon$ & $\text{poly}(\log n,1/\epsilon)$  & \xmark & \cmark & w.h.p & \textbf{This work}\textsuperscript{\ref{charikar-slomon-footnote}} \bigstrut \\ \hline		

				$2$ & $O((m+n)^{1/\sqrt2})$  & \cmark & \xmark & & Ivkovi\'{c} and Lloyd (WG '93) \cite{ivkovic1994fully}\bigstrut \\ 		\hline
				
				$2$ & $O(\log n)$  & \xmark & \xmark & $O(\log^2n)$ w.h.p & Baswana et al. (FOCS '11) \cite{baswana2015fully}\bigstrut \\ 		\hline
				
				$2$ & $O(1)$  & \xmark & \xmark & w.h.p & Solomon (FOCS '16) \cite{solomon2016fully}\bigstrut \\ 	
				\specialrule{.125em}{.0625em}{.0625em}				 				
				
				$3/2 +\epsilon$ & $O(\sqrt[4]{m}/ \epsilon^{2.5})$ & \cmark & \cmark & bipartite only & Bernstein and Stein (ICALP '15) \cite{bernstein2015fully}\bigstrut \\ 		\hline
				
				$3/2 +\epsilon$ & $O(\sqrt[4]{m}/ \epsilon^{2.5})$ & \cmark & \xmark & 
				& Bernstein and Stein (SODA '16) \cite{bernstein2016faster} \bigstrut \\ 		\hline
				
				$3/2$ & $O(\sqrt{m})$ & \cmark & \cmark &  & Neiman and Solomon (STOC '13) \cite{neiman2016simple} \bigstrut \\ 		\hline

				$1+\epsilon$ & $O(\sqrt{m}/\epsilon^{2})$  & \cmark & \cmark & & Gupta and Peng (FOCS '13) \cite{gupta2013fully}\bigstrut \\ 		\hline
			\end{tabular}
			\end{center}}	
			
		\captionsetup{justification=centering}
		\caption{Our Results and Previous Results for Fully-Dynamic Matching.}	
		\label{table:Results}	
					
		\end{table}
	\vspace{-0.5cm}
		{(All references are to the latest publication, with the first publication venue in parentheses.)}
	\end{center}
	
%

	\subsection{Our Contribution}
	
	Our main technical result requires the following natural definition of $(c,d)$-approximately-maximal fractional matchings.
	
	\begin{Def}[Approximately-Maximal Fractional Matching]\label{def:approx-max}
		We say that a fractional matching $w:E\rightarrow \mathbb{R}^+$ is \emph{$(c,d)$-approximately-maximal} if  every edge $e\in E$ either has fractional value $w_e\geq 1/d$ or one endpoint $v$ with sum of incident edges' weights at least $W_v\triangleq \sum_{e\ni v} w_e\geq 1/c$ and moreover all edges $e'$ incident on $v$ have $w_{e'} \leq  1/d$.
	\end{Def}
	
	Note that this definition generalizes maximal fractional matchings (for which $c=d=1$). The second condition required of $v$ above (i.e., having no incident edges $e'$ with $w_{e'}> 1/d$) may seem a little puzzling, but will prove important later; it can be safely ignored until \S\ref{sec:partition} and \S\ref{sec:sampling}.
	
	Our main qualitative result, underlying our quantitative result, is the following black-box reduction from integral matching algorithms to approximately-maximal fractional matching algorithms, as stated in the following theorem.
	
	\begin{restatable}{thm}{reduction}\label{thm:reduction}
		Let $\mathcal{A}$ be a fully-dynamic  $(c,d)$-approximately-maximal fractional matching algorithm  with update time $T(n,m)$ which changes at most $C(n,m)$ edge weights per update, for some $d\geq \frac{9c\log n}{\epsilon^2}$, with $\epsilon\leq \frac{1}{2}$. Then, there exists a randomized fully-dynamic $(2c+O(\epsilon))$-approximate \emph{integral} matching algorithm $\mathcal{A}'$ with update time $T(n,m)+O(C(n,m)\cdot d/\epsilon^2)$ with high probability.
		Moreover, if $T(n,m)$ and $C(n,m)$ are worst-case bounds, so is the update time of Algorithm $\mathcal{A}'$.
	\end{restatable}

	
	Now, one may wonder whether fully-dynamic $(c,d)$-approximately-maximal fractional matching algorithms with low worst-case update time and few edge weight changes exist for any non-trivial values of $c$ and $d$. Indeed, the recent algorithm of \citet{bhattacharya2017fully} is such an algorithm, as the following lemma asserts.
	
	\begin{restatable}[\cite{bhattacharya2017fully}]{lem}{fractwopolylog}\label{lem:two-polylog}
		For all $\epsilon\leq \frac{1}{2}$, there exists a fully-dynamic $(1+2\epsilon, \max\{54\log n/\epsilon^3,(3/\epsilon)^{21}\})$-approximately-maximal fractional matching algorithm with  $T(n,m)=O(\log^3n/\epsilon^7)$ worst-case update time, using at most $C(n,m)=O(\log n/\epsilon^2)$ edge weight changes per update in the worst case.
	\end{restatable}

	We highlight the general approach of the algorithm of \citet{bhattacharya2017fully} in \S\ref{sec:partition} to substantiate the bounds given in \Cref{lem:two-polylog}.
	Plugging the values of $c$, $T(n,m)$ and $C(n,m)$ of \Cref{lem:two-polylog} into \Cref{thm:reduction} immediately yields our result, given in the following theorem.\footnote{We note that previously and independently from \cite{bhattacharya2017fully}, we
obtained similar results to those of Theorem \ref{2-polylog}. After receiving and reading a preprint of \cite{bhattacharya2017fully}, we realized that using \cite{bhattacharya2017fully} and our simple reduction we can obtain this theorem in a much simpler way, and therefore present only this much simpler algorithm.}
	
	\begin{thm}\label{2-polylog}
		For all $\epsilon\leq \frac{1}{2}$, there exists a randomized fully-dynamic $(2+O(\epsilon))$-approximate integral
		matching algorithm with $O(\log^3 n /\epsilon^7 +  \log n/\epsilon^{2}\cdot \max\{\log n/\epsilon^3,(3/\epsilon)^{21}\}) = O_{\epsilon}(\log^3 n)$ worst-case update time.
	\end{thm}
%

	We recall that until now, for worst-case polylog update times only fractional algorithms -- algorithms which only approximate the \emph{value} of the maximum matching -- were known for this problem.
	
	Finally, combined with the recent black-box reduction of \citet{stubbs2017metatheorems} from the weighted to the unweighted matching problem, our algorithm also yields the first fully-dynamic constant-approximate maximum weight matching algorithm with polylogarithmic worst-case update time. 
	
	\begin{restatable}{thm}{mwm}\label{thm:4-mwm}
		For all $\epsilon\leq \frac{1}{2}$, there exists a randomized fully-dynamic $(4+O(\epsilon))$-approximate maximum weight
		matching algorithm with $\text{poly}(\log n,1/\epsilon)$ worst-case update time.
	\end{restatable}
%
	
	
	\subsection{Our Techniques}
	Our framework yielding our main result combines three ingredients: approximately-maximal fractional matchings, kernels and fast $(1+\epsilon)$ matching algorithms for bounded-degree graphs. We give a short exposition of these ingredients and conclude with how we combine all three.
	
	\paragraph{Approximately-Maximal Fractional Matchings.} The first ingredient we rely on is $(c,d)$-approximately-maximal fractional matchings, introduced in the previous section. Recalling that for such solutions, each edge has value at least $1/d$ or one of its endpoints has sum of incident  edge values at least $1/c$. This approximate maximality condition implies this fractional matching has high value compared to the maximum matching size; specifically, this fractional matching's size is at least a $1/2\max\{c,d\}$ fraction of this value (easily verifiable using LP duality). As we shall show, approximate maximality also allows one to use these fractional values to sample a subgraph in the support of this fractional matching which contains a large \emph{integral} matching compared to $G$, with high probability. We discuss the dynamic fractional matching algorithm of \citet{bhattacharya2015deterministic} and show that it maintains an approximately-maximal fractional matching in \S\ref{sec:partition}.
	
	\paragraph{Kernels.} The second ingredient we rely on is the notion of \emph{kernels}, introduced by \cite{bhattacharya2015deterministic}. Roughly speaking, a kernel is a low-degree subgraph $H$ of $G$ such that each edge of $G$ not taken into $H$ has at least one endpoint whose degree in $H$ is at least $1/c$ times the maximum degree in $H$. Relying on Vizing's Theorem \cite{vizing1964estimate}, we show in \S\ref{sec:kernel} that such a graph has maximum matching size $\mu(H)$ at least $1/(2c+\epsilon)$ of the matching size of $G$, previously only known for kernels of bipartite graphs, where this is easily verifiable via LP duality.\footnote{As a byproduct of our proof, we show how the algorithms of \citet{bhattacharya2015deterministic} can be made $(2+\epsilon)$-approximate within the same time bounds. As this is tangential to our main result, we do not elaborate on this.} Efficiently maintaining a large matching can therefore be reduced to maintaining a low-degree kernel, given the last ingredient of our approach.
	
	\paragraph{Bounded-Degree $\boldsymbol{(1+\epsilon)}$-matching.} The final ingredient we rely on for our framework is $(1+\epsilon)$ matching algorithms with worst-case update time bounded by the graph's maximum degree, such as the algorithms of \citet{gupta2013fully,peleg2016dynamic}.
	
	\paragraph*{Our approach in a nutshell.} Given the above ingredients, our framework is a simple and natural one. Throughout our algorithm's run, we run a fully-dynamic $(c,d)$-approximately-maximal fractional matching algorithm  with efficient worst-case update. Sampling edges independently according to this fractional value (times some logarithmic term in $n$, to guarantee concentration) allows us to sample a kernel of logarithmic maximum degree, with each non-sampled edge having at least one endpoint with degree at least $1/c$ times the maximum subgraph degree, with high probability. As the obtained subgraph $H$ therefore has a maximum matching of size at least $\approx 1/2c$ times the maximum matching in $G$, a $(1+\epsilon)$-matching algorithm in $H$ yields a $\approx 2c+O(\epsilon)$ matching in $G$. We then maintain a $(1+\epsilon)$-matching in $H$ (which by virtue of $H$'s bounded degree we can do in logarithmic worst-case time) following each update to $H$ incurred by a change of some edge's fractional value by the dynamic fractional matching algorithm. The obtained integral algorithm's update time is dominated by two terms: the running time of the fractional algorithm, and the number of edge weight updates per update, times $O(\log n)$. This concludes the high-level analysis of the obtained approximation ratio and update time of our approach, as given in \Cref{thm:reduction}.
	
	\paragraph*{Wider applicability.} We stress that our framework is general, and can use any approximately-maximal fractional matching algorithm. Consequently, any improvement on the running time and number of edge value changes for maintaining approximately-maximal fractional matchings yields a faster worst-case update time.
	Likewise, any approximately-fractional matching algorithm which maintains a ``more maximal'' fractional solution yields better approximation ratios. 
\section{Preliminaries}
In this section we introduce some previous results which we will rely on in our algorithm and its analysis. We start by reviewing the approach of \citet{bhattacharya2017fully} to obtain efficient fractional algorithms in \S\ref{sec:partition}. We then discuss the bounded-degree subgraphs we will consider, also known as \emph{kernels}, in \S\ref{sec:kernel}. Finally, we briefly outline the $(1+\epsilon)$-approximate $O(\Delta/\epsilon^2)$ worst case update time algorithms we will rely on for our algorithm, in \S\ref{sec:bounded-deg}.

\subsection{Hierarchical Partitions}\label{sec:partition}
In this section we review the approximately-maximal fractional matchings maintained by \citet{bhattacharya2017fully}.
At a high level, this algorithm relies on the notion \emph{hierarchical partitions}, in which vertices are assigned some \emph{level} (the partition here is given by the level sets), and edges are assigned a fractional value based on their endpoints' levels. Specifically, an edge is assigned a value exponentially small in its vertices' maximum level. The levels (and therefore the edge weights) are updated in a way as to guarantee feasibility, as well as guaranteeing that a vertex $v$ of high level has high sum of incident edge weights, $W_v$. These conditions are sufficient to guarantee approximate maximality, as we shall soon show.

The hierarchical partitions considered by \citet{bhattacharya2017fully}, termed simply \emph{nice partitions}, is described as follows. In the definition constants $\beta, K, L$ and a function $f(\beta)=1-3/\beta$ are used, satisfying the following.
\begin{equation}\label{eqn:constants}
	\beta \geq 5,\, K = 20,\, f(\beta) = 1-3/\beta,\, L = \lceil\log_\beta n\rceil.
\end{equation}
In our case, for some $\epsilon\leq \frac{1}{2}$, we will let $\beta = \frac{3}{\epsilon} (\geq 5)$.
As we will be shooting for $O(1)$-approximation algorithms with polylogarithmic update time and our reduction's update time has polynomial dependence on $\epsilon^{-1}$, we will assume without loss of generality that $\epsilon = \Omega( \frac{3}{\sqrt[20]{n}})$, and so for $n$ large enough, we have $K\leq L$.

\begin{Def}[A nice partition \cite{bhattacharya2017fully}]\label{def:nice-partition}
	In a \emph{nice partition} of a graph $G=(V,E)$, each vertex $v$ is assigned an integral \emph{level} $\ell(v)$ in the set $\{K,K+1,\dots,L\}$.  In addition, for each vertex $v\in V$ and edge $e\ni v$ the \emph{shadow-level of $v$ with respect to $e$}, denoted by $\ell_v(e)$, is a (positive) integer satisfying  $\ell(v)-1\leq \ell_v(e)\leq \ell(v)+1$. Moreover, for each vertex $v$, we have 
	\begin{equation}\label{eqn:shadow}
		\max_{e\ni v} \ell_v(e) - \min_{e\ni v} \ell_v(e) \leq 1.
	\end{equation}		
	The level of an edge $e=(u,v)$ is taken to be the maximum shadow-level of an endpoint of $e$ with respect to $e$; i.e., $\ell(u,v)=\max\{\ell_u(e), \ell_v(e)\}$. Let $W_v=\sum_{e\in v} w_e$ be the sum of weights of edges incident on a vertex $v$. Then,
	\begin{enumerate}
		\item\label{prop:edges} For every edge $e\in E$, it holds that $w_e =\beta^{-\ell(e)}$.
		\item \label{prop:ub} For every node $v\in V$, it holds that $W_v<1$.
		\item \label{prop:lb} For every node $v\in V$ with level $\ell(v)>K$, it holds that $W_v\geq f(\beta)$.
	\end{enumerate}
\end{Def}
The intuition behind this definition in \citet{bhattacharya2017fully} is to mimic the hierarchical partition of \citet{bhattacharya2015deterministic}, termed  \emph{$(\alpha,\beta)$-decompositions} there. $(\alpha,\beta)$-decompositions are the special case of nice partitions where the shadow-level of a vertex $v$ with respect to each edge $e\ni v$ is precisely equal to the vertex's level; i.e, $\ell_v(e)= \ell(v)$ (with  $\alpha$ denoting $f(\beta)/\beta$). The advantage of this more relaxed notion of shadow-level is to allow a vertex to move between levels ``slowly'', only notifying \emph{part} of its incident edges of its level change between updates, and therefore only updating some of its edges' weights. This allows for maintaining this partition with fast worst-case update time, as shown in \citet{bhattacharya2017fully} (more on this below).

This above intuition concerning nice partitions will not prove important for our analysis. The crucial property we will rely on is given by the following lemma, which asserts that the fractional matching associated with a nice partition is approximately-maximal.



\begin{lem}\label{lem:nice-approx}
	Let $\epsilon\leq \frac{1}{2}$. Consider a nice partition with $\beta=3/\epsilon\geq 6\geq 5$, and so $f(\beta)=1-\epsilon$. Then, the fractional matching associated with this nice partition is $(1+2\epsilon, \max\{54\log n/\epsilon^3,(3/\epsilon)^{21}\})$-approximately-maximal.
\end{lem}
\begin{proof}
	Let $K'=\max\{\lceil \log_\beta (18c\log n/\epsilon^2)\rceil, K+1 \}$, and let $d=\beta^{K'}\leq \max\{54\log n/\epsilon^3,(3/\epsilon)^{21}\}$.
	For any edge $e$, 
	if $\ell(e) = \max_{v\in e} \{\ell_v(e)\} \leq K'$,  
	then by definition $w_e = \beta^{-\ell(e)} \geq \beta^{-K'} = \frac{1}{d}$. Alternatively, if $w_e < \frac{1}{d}$ then $\ell(e) > K'$ and therefore by integrality of $\ell(e)$, we have $\ell(e)\geq K'+1$. Now, let $v$ be $\arg\max_{v\in e} \{\ell_v(e)\} \geq K'+1$. Then, by definition of shadow-levels and $K'>K$, we have $\ell(v)\geq \ell_v(e) - 1 \geq K' > K$ and so by Property \ref{prop:lb} of a nice partition we have $W_v>f(\beta) = 1-\epsilon \geq \frac{1}{1+2\epsilon}$ (as $\epsilon\leq \frac{1}{2}$). But on the other hand, by \Cref{eqn:shadow}, we also know that for every edge $e'\ni v$, 
$$
		\ell_v(e') \geq \min_{e'\ni v} \ell_v(e') \geq \max_{e'\ni v} \ell_v(e') - 1 \geq \ell_v(e) - 1 \geq K' > K.
$$
	Therefore, by definition of the edge weights, each edge $e'\ni v$ satisfies $w_{e'} \leq \beta^{-K'} \leq \frac{1}{d}$. 
\end{proof}

The recent result of \citet{bhattacharya2017fully} for maintaining nice partitions in worst-case update time together with \Cref{lem:nice-approx} immediately implies \Cref{lem:two-polylog}, restated below. We substantiate these bounds with the dependence on $\epsilon$ stated explicitly in \S\ref{sec:justification}, as \citet{bhattacharya2017fully} had $\epsilon=O(1)$ and so their results do not state these dependencies explicitly.

\fractwopolylog*

As we shall show, approximately-maximal fractional matchings allow us to sample a bounded-degree subgraph $H$ of $G$ containing a large matching compared to the maximum matching size in $G$, $\mu(G)$. For this we will require the notion of \emph{kernels}, defined in \S\ref{sec:kernel}.

\subsection{Kernels}\label{sec:kernel}
In this section we review the concept of kernels, first introduced by \citet{bhattacharya2015deterministic}.

\begin{Def}[Kernels \cite{bhattacharya2015deterministic}]
	A \emph{$(c,d)$-kernel} of a graph $G$ is a subgraph $H$ of $G$ satisfying:
	\begin{enumerate}
		\item \label{p1:bounded-deg} For each vertex $v\in V$, the degree of $v$ in $H$ is at most $d_H(v)\leq d$.
		\item \label{p2:satisfied-edges} For each edge $(u,v)\in E\setminus H$, it holds that $\max\{d_H(u),d_H(v)\}\geq d/c$.
	\end{enumerate}
\end{Def}

The interest in finding a bounded-degree subgraph $H$ of $G$ may seem natural, as one may expect to be able to compute a matching quickly in $H$ due to its sparsity (we elaborate more on this point in \S\ref{sec:bounded-deg}).
The interest in satisfying the second property, on the other hand, may seem a little cryptic. However, combining both properties implies that the matching number of $H$, $\mu(H)$, is large in comparison with the matching number of $G$, $\mu(G)$.

\begin{lem}\label{lem:kernel-matching-number}
	Let $H$ be a $(c, d)$-kernel of $G$. Then $\mu(H)\geq \frac{1}{2c(1+1/d)}\cdot \mu(G)$.
\end{lem}
\begin{proof}
	Consider the following fractional matching solution, 
	$$
	f_{u,v} =
	\begin{cases}
	\frac{1}{d}	& (u,v)\in H\setminus M^* \\
	\max\{1-\frac{d_h(u)+d_H(v)-2}{d},0\} & (u,v)\in H\cap M^*.
	\end{cases}
	$$
	This is a feasible fractional matching due to the degree bound of $H$ and the fractional values assigned to edges of a vertex $v$ incident on an edge  $e\in H\cap M^*$ being at most $\frac{d_H(u)-1}{d}+\frac{d-d_H(v)+1}{d}=1$.
	To show that this fractional matching has high value, consider the variables $y_v = \sum_u f_{u,v}$. On the one hand, by the handshake lemma, $\sum_{u,v} f_{u,v} = \frac{1}{2}\sum_v y_v$. On the other hand, each edge of $(u,v)\in M^*\cap H$ has $y_u + y_v \geq 1\geq \frac{1}{c}$ by construction and each edge of $M^*\setminus H$ has at least one endpoint $v$ of degree $d_H(v)\geq \frac{d}{c}$, implying that $y_u+y_v \geq \frac{d}{c}\cdot \frac{1}{d} = \frac{1}{c}$ for each $(u,v)\in M^*\setminus H$. As each vertex $v$ neighbors at most one edge of $M^*$, we obtain
	$$
	\sum_e f_e = \frac{1}{2}\cdot \sum_v y_v \geq \frac{1}{2c}\cdot |M^*| =  \frac{1}{2c}\cdot
	\mu(G). 
$$	
	Now, to show that $H$ contains a large \emph{integral} matching, we rely on Vizing's Theorem \cite{vizing1964estimate}, which asserts that every multigraph of maximum degree $\Delta$ and maximum edge multiplicity $\mu$ has a proper $\Delta+\mu$ edge-coloring; i.e., a partition of the edge set into $\Delta+\mu$ edge-disjoint matchings. To use this theorem, we construct a multigraph on the same vertex set $V$ with each edge $e$ replaced by $f_e\cdot d$ parallel copies (note that $f_e\cdot d$ is integral). By construction, the number of edges in this multigraph is $\sum_e f_e \cdot d$. By feasibility of $f$, we have that this multigraph has maximum degree $d$. By Vizing's Theorem, the simple subgraph obtained by ignoring parallel edges corresponding to edges in $H\cap M^*$ can be edge colored using $d+1$ colors. But for each edge $e=(u,v)\in H\cap M^*$, such a coloring uses at most $d_H(u)-1+d_H(v)-1$ distinct colors incident on $u$ or $v$. To extend this coloring to a proper coloring of the multigraph, we color the  $\max\{d-(d_H(u)-1+d_H(v)-1),0\}$ multiple edges $(u,v)$ in this multigraph using some $\max\{d-(d_H(u)-1+d_H(v)-1),0\}$ unused colors of the palette of size $d+1$ used so far. We conclude that the support of this multigraph (i.e., the graph $H$), which has $\sum_e f_e \cdot d$ edges, contains a matching of size at least 
	$$
	\frac{1}{d+1}\cdot \sum_e f_e \cdot d  = \frac{1}{1+1/d}\cdot \sum_e f_e \geq \frac{1}{2c(1+1/d)} \cdot \mu(G). \qedhere
	$$
\end{proof}

\Cref{lem:kernel-matching-number} and the algorithm of \S\ref{sec:bounded-deg} immediately imply that the algorithms of \citet{bhattacharya2015deterministic} can be made $(2+\epsilon)$-approximate within the same time bounds (up to $\text{poly}(1/\epsilon)$ terms). As this was previously also observed in \citet{bhattacharya2016new}, we do not elaborate on this point here.

\subsection{Nearly-Maximum Matchings in Degree-Bounded Graphs}\label{sec:bounded-deg}

In this short subsection we highlight one final component we will rely on for our algorithm: fast nearly-optimal matching algorithms with worst-case update time bounded by $G$'s maximum degree. Such algorithms were given by \citet{gupta2013fully,peleg2016dynamic}. More precisely, we have the following lemma. The bound for the algorithm of \citet{peleg2016dynamic} follows as $\alpha\leq \Delta$ always, while the bound for the algorithm of \citet{gupta2013fully} is immediate by inspecting this algorithm, as observed in \cite{peleg2016dynamic}.

\begin{lem}[\cite{gupta2013fully,peleg2016dynamic}]\label{lem:bounded-deg}
	There exists a dynamic $(1+\epsilon)$-approximate matching algorithm with worst-case $O(\Delta/\epsilon^2)$ update time in dynamic graphs whose maximum degree is always at most $\Delta$.
\end{lem}
\section{Sampling Using Approximately-Maximal Matchings}\label{sec:sampling}
In what follows we will show that sampling edges independently with probability roughly proportional to their assigned value according to an approximately-maximal fractional matching yields a kernel of logarithmic maximum degree with high probability.

\begin{lem}\label{lem:sample-kernel}
	Let $\epsilon\leq \frac{1}{2}$. Let $w:E\rightarrow \mathbb{R}^+$ be a $(c,d)$-approximately-maximal fractional matching with $d\geq \frac{9c\log n}{\epsilon^2}$. Then, sampling each edge $e$ independently with probability 
	\begin{equation}\label{eqn:sample}
	\min\{1,w_e\cdot d\}
	\end{equation}
	yields a subgraph $H$ which is a $(c(1+O(\epsilon)), O(d))$-kernel of $G$ with high probability.
\end{lem}
\begin{proof}
	For any vertex $v\in V$, denote by $D_v$ the random variable which corresponds to $v$'s degree in $H$. As before, denote by $W_v = \sum_{e\ni v} w_e \leq 1$ the sum of edge weights of edges incident on $v$. 
	
	First, we prove the degree upper bound; i.e., Property \ref{p1:bounded-deg} of a kernel. As $w$ is a fractional matching, we know that $W_v \leq 1$. Therefore, by \Cref{eqn:sample}, we have that $\mathbb{E}[D_v]\leq d$. By standard Chernoff bounds, as $d\geq 9c\cdot \log n/\epsilon^2$, we find that 
	$$
	\Pr[D_v\geq (1+\epsilon) \cdot d]  \leq \exp\left(\frac{-\epsilon^2 \cdot d}{3}\right) \leq \frac{1}{n^{3c}}\leq \frac{1}{n^3}.
	$$
	
	Next, we prove that any edge not sampled into $H$ will, with high probability, be incident on some high-degree vertex in $H$; i.e., we show that $H$ satisfies Property \ref{p2:satisfied-edges} of a kernel. 
	First, note that an edge $e$ with $w_e\geq 1/d$ will be sampled with probability one, given our sampling probability given in \Cref{eqn:sample}, therefore trivially satisfying Property \ref{p2:satisfied-edges} of a kernel. Conversely, an edge $e$ with $w_e<1/d$ has some endpoint $v$ with $W_v\geq 1/c$ and all edges $e'$ incident on $v$ have $w_{e'}\leq 1/d$, since $w$ is $(c,d)$-approximately maximal. Therefore, by \Cref{eqn:sample} each edge $e'$ incident on $v$ is sampled with probability precisely $w_{e'}\cdot d$. Consequently, we have that $\mu_v = \mathbb{E}[D_v] = W_v \cdot d \geq d/c$. By standard Chernoff bounds, as $\mu_v \geq d/c\geq \frac{9\log n}{\epsilon^2}$, we find that 
	$$
		\Pr[D_v\leq (1-\epsilon) \cdot d/c]  \leq 
		\Pr[D_v\leq (1-\epsilon)\cdot \mu_v] \leq \exp\left(\frac{-\epsilon^2 \cdot \mu_v}{2}\right) \leq \frac{1}{n^{4.5}}.
	$$
	Taking a union bound over the $O(n^2)$ possible bad events corresponding to violating a property of a $(c(1+\epsilon)/(1-\epsilon)),d(1+\epsilon))$-kernel, we find that with high probability 
	\begin{enumerate}
		\item For each vertex $v\in V$, it holds that $d_H(v)\leq (1+\epsilon)\cdot d$.
		\item For each edge $(u,v)\in E\setminus H$, it holds that $\max\{d_H(u),d_H(v)\}\geq (1-\epsilon)\cdot d/c$.
	\end{enumerate}
	In other words, $H$ is a $(c(1+\epsilon)/(1-\epsilon), d(1+\epsilon))$-kernel of $G$  with high probability.
\end{proof}

\section{Our Reduction}
Given the previous sections, we are now ready to describe our reduction from fully-dynamic integral matching to approximately-maximal fractional matching and analyzing its performance, given by \Cref{thm:reduction}, restated here.

\reduction*

\begin{proof}
	Our reduction works as follows. Whenever an edge $e$ is added/removed from $G$, we update the $(c,d)$-approximately-maximal fractional matching, using algorithm $\mathcal{A}$, in time $T(n,m)$. We then sample each of the at most $C(n,m)$ edges $e$ whose value is changed, independently, with  probability given by \Cref{eqn:sample}. 
	To control the maximum degree in the sampled subgraph $H$, every vertex $v$ maintains a list of at most $(1+\epsilon)\cdot d$ sampled edges ``allowable'' for use in $H$. (This list can be maintained dynamically in a straightforward manner by maintaining the list of all sampled edges of $v$ and having the shorter list consist of the first $(1+\epsilon)d$ sampled edges of $v$.) We let $H$ be the graph induced by the sampled edges ``allowed'' by both their endpoints.
	Finally, we use a $(1+\epsilon)$-matching algorithm as in \Cref{lem:bounded-deg} to maintain a matching in the sampled subgraph $H$.
	
	By \Cref{lem:sample-kernel} the subgraph $H$ is a $(c(1+O(\epsilon)),d)$-kernel of $G$ with high probability (note that by the same lemma, all sampled edges will appear in our $H$). By \Cref{lem:kernel-matching-number}, this means that with high probability this kernel $H$ has matching number at least  
	$$
		\mu(H)  \geq \frac{1}{2c(1+O(\epsilon))(1+1/d))} \cdot \mu(G) \geq \frac{1}{2c+O(\epsilon)}\cdot \mu(G),
	$$ 
	where the second inequality follows from $d\geq \frac{9c\log n}{\epsilon^2} \geq \frac{1}{\epsilon}$.
	Therefore, a $(1+\epsilon)$-approximate matching in $H$ is itself a $(2c+O(\epsilon))$-approximate matching in $G$. Now, each of the $C(n,m)$ changes to edge weights of the fractional matching incurs at most three updates to the kernel $H$: for every edge $e$ whose weight changes, this edge can be added/removed to/from $H$ if it is sampled in/out; in the latter case both of $e$'s endpoints can have a new edge added to their ``allowable'' edge list in place of $e$, and therefore possibly added to $H$, in case the endpoints had more than $(1+\epsilon)d$ sampled edges.
	But on the other hand, the $(1+\epsilon)$-approximate matching algorithm implied by \Cref{lem:bounded-deg} requires $O(d/\epsilon^2)$ worst-case time per update in $H$, by $H$'s worst-case degree bound. Consequently, our algorithm maintains a $(2c+O(\epsilon))$-approximate integral matching w.h.p in $T(n,m) + O(C(n,m)\cdot d/\epsilon^2)$ update time; moreover, this update time is worst case if the bounds on $T(n,m)$ and $C(n,m)$ are both worst case.
\end{proof}
\vspace{-0.2cm}
\section{Applications to Maximum Weight Matchings}\label{sec:mwm}
In this section we highlight the consequences of our results for fully-dynamic maximum weight matching. First, we discuss a new reduction of \citet{stubbs2017metatheorems}.

\begin{lem}[\cite{stubbs2017metatheorems}]\label{lem:mwm-reduction}
	Let $\mathcal{A}$ be an fully-dynamic $\alpha$-approximate maximum cardinality matching algorithm with update time $T(n,m)$. Then, there exists a fully-dynamic $2\alpha(1+\epsilon)$-approximate maximum cardinality matching algorithm with update time $O(T(n,m)\cdot \frac{\log(n/\epsilon)}{\epsilon})$.
	Furthermore, if Algorithm $\mathcal{A}$ is deterministic, so is the new one, and if Algorithm $\mathcal{A}$'s update time is worst case, so is the new algorithm's update time.
\end{lem}

This reduction (which we elaborate on shortly), together with the state of the art dynamic maximum matching algorithms, implies most of the best currently best bounds for dynamic maximum weight matching, in \Cref{table:mwm} below.
\vspace{-0.25cm}

\begin{center}
	\begin{table*}[ht]
		\centering
		\small
		\begin{tabular}{ | c | c | c | c | c |}
			\hline
			Approx. & Update Time & det. & w.c.  & reference \\  \hline 
			
			$4+\epsilon$ & $O(\log(n/\epsilon)/\epsilon)$ \bigstrut & \xmark & \xmark & \ref{lem:mwm-reduction} + Solomon (FOCS '16) \cite{solomon2016fully}\bigstrut\\ \hline
			
			$4+\epsilon$ & $\text{poly}(\log n,1/\epsilon)$  & \cmark & \xmark & \ref{lem:mwm-reduction} + Bhattacharya et al.~(STOC '16) \cite{bhattacharya2016new} \\ \hline
			
			$4+\epsilon$ & $O(m^{1/3}\log(n/\epsilon)/\epsilon^{3})$ & \cmark & \cmark & \ref{lem:mwm-reduction} + Bhattacharya et al.~(SODA '15) \cite{bhattacharya2015deterministic}\bigstrut \\ 		\hline
			
			$4+\epsilon$ & $\text{poly}(\log n,1/\epsilon)$ \bigstrut & \xmark & \cmark & \ref{lem:mwm-reduction} + \textbf{This work} \bigstrut\\ \hline
			
			$\approx 3.009 +\epsilon$ & $O(\sqrt{m}\log C \epsilon^{-3})$  & \cmark & \cmark & Gupta and Peng (FOCS '13) \cite{gupta2013fully}\bigstrut\\
			\hline
			
			$3 +\epsilon$ & $O(\sqrt[4]{m}\log (n/\epsilon) \epsilon^{-3})$  & \cmark & \xmark & \ref{lem:mwm-reduction} + Bernstein and Stein (SODA '16) \cite{bernstein2016faster} \bigstrut\\
			\hline

			$2+\epsilon$ & $O(\sqrt{m}\cdot \frac{\log^2 (n/\epsilon)} {\epsilon^{4}})$  & \cmark & \cmark &  \ref{lem:mwm-reduction} + Gupta and Peng (FOCS '13) \cite{gupta2013fully}\bigstrut\\
			\hline
			
			$1 +\epsilon$ & $O(\sqrt{m} C \epsilon^{-3})$  & \cmark & \cmark & Gupta and Peng (FOCS '13) \cite{gupta2013fully}\bigstrut\\
			\hline
			
			$1 +\epsilon$ & $O(\sqrt{m}\log C \epsilon^{-O(1/\epsilon)})$  & \cmark & \cmark & Gupta and Peng (FOCS '13) \cite{gupta2013fully}\bigstrut\\
			\hline	
		\end{tabular}
		\captionsetup{justification=centering}
		\caption{Our Results and Previous State-of-the-Art for Fully-Dynamic MWM.}	
		\label{table:mwm}	
	\end{table*}
\end{center}

\vspace{-0.8cm}

A somewhat more involved and worse update time bound than that given in \Cref{lem:mwm-reduction} was presented in \cite{stubbs2017metatheorems}, as that paper's authors sought to obtain a persistent matching, in a sense that this matching should not change completely after a single step (i.e., no more than $O(T(n,m))$ changes to the matching per edge update, if $T(n,m)$ is the algorithm's update time). However, a simpler and more efficient reduction yielding a non-persistent matching algorithm with the performance guarantees of \Cref{lem:mwm-reduction} is implied immediately from the driving observation of \citet{stubbs2017metatheorems} (and indeed, is discussed in \cite{stubbs2017metatheorems}). This observation, previously observed by \citet{crouch2014improved} in the streaming setting, is as follows: denote by $E_i$ the edges of weights in the range $((1+\epsilon)^i,(1+\epsilon)^{i+1}]$, and let $M_i$ be an $\alpha$-approximate matching in $G[E_i]$. Then, greedily constructing a matching by adding edges from each $M_i$ in decreasing order of $i$ yields a $2\alpha(1+\epsilon)$-approximate maximum weight matching. Adding to this observation the observation that if we are content with a $(1+\epsilon)$-approximate (or worse) maximum weight matching  we may safely ignore all edges of weight less than $\epsilon/n$ of the maximum edge weight (a trivial lower bound on the maximum weight matching's weight), we find that we can focus on the ranges $((n/\epsilon)^i,(n/\epsilon)^{i+2}]$, for some $i\in \mathbb{Z}$, noting that each edge belongs to at most two such ranges. 

In each such range $((n/\epsilon)^i,(n/\epsilon)^{i+2}]$, the argument of \cite{crouch2014improved,stubbs2017metatheorems} implies that maintaining $\alpha$-approximate matchings in the sub-ranges $((1+\epsilon)^j,(1+\epsilon)^{j+1}]$ for integral ranges and combining these greedily result in a $2\alpha(1+\epsilon)$-approximate maximum weight matching in the range $((n/\epsilon)^i,(n/\epsilon)^{i+2}]$. Therefore, in the range containing a $(1+\epsilon)$-approximate MWM (such a range exists, by the above), this approach maintains a $2\alpha(1+O(\epsilon))$-approximate MWM. The only possible difficulty is combining these matchings greedily \emph{dynamically}. This is relatively straightforward to do in $O(\log (n/\epsilon)/\epsilon)$ worst-case time per change of the $\alpha$-approximate matching algorithm, however, implying the bound of \Cref{lem:mwm-reduction}.


As seen in \Cref{table:mwm}, this reduction of \citet{stubbs2017metatheorems} implies a slew of improved bounds for fully-dynamic approximate maximum weight matching. Plugging in our bounds of \Cref{2-polylog} for fully-dynamic maximum matching into the reduction of \Cref{lem:mwm-reduction} similarly yields the first constant-approximate maximum weight matching with polylogarithmic \emph{worst-case} update time, given in \Cref{thm:4-mwm} below.

\mwm*

\section{Conclusion and Future Work}
In this work we presented a simple randomized reduction from $(2c+\epsilon)$-approximate fully-dynamic matching to fully-dynamic $(c,d)$-approximately-maximal fractional matching with a slowdown of $d$. Using the recent algorithm of \citet{bhattacharya2017fully}, our work yields the first fully-dynamic matching algorithms with faster-than-polynomial worst-case update time for any constant approximation ratio; specifically, it yields a $(2+O(\epsilon))$-approximate matching with polylog update time. Our work raises several natural questions and future research directions to explore.

\paragraph*{Faster/``More Maximal'' Fractional Algorithms.} Our reduction yields $O(1)$-approximate algorithms with polylogarithmic update times whose approximation ratios and update time are determined by the  $(c,d)$-approximately-maximal fractional matching algorithms they rely on. There are two venues to pursue here: the first, in order to improve the update time of our $(2+\epsilon)$-approximate matching algorithm, would be to improve the update time of fully-dynamic $(1+\epsilon, \max\{18\log n/\epsilon^2,(3/\epsilon)^{20}\})$-approximately-maximal fractional matching algorithm of  \citet{bhattacharya2017fully}. Generally, faster $(c,d)$-approximately maximal fractional matching algorithms would imply faster randomized $(2c+\epsilon)$-approximate matching integral algorithms.

\paragraph*{More Efficient Reduction.} Another natural question is whether the dependence on $d$ may be removed from our reduction, yielding randomized integral matching algorithms with the same running time as their fractional counterparts.

\paragraph*{Deterministic Reduction.} A further natural question would be whether or not one can obtain a \emph{deterministic} counterpart to our black-box reduction from integral matching to approximately-maximal fractional matching. Any such reduction with polylogarithmic update time would yield deterministic algorithms with worst-case polylogarithmic update times.

\paragraph*{Maximal Matching.} Finally, a natural question from our work and prior work is whether or not a \emph{maximal} matching can be maintained in worst-case polylogarithmic time (also implying a $2$-approximate minimum vertex cover within the same time bounds). We leave this as a tantalizing open question.

\subsection*{Acknowledgments.} We thank Monika Henzinger for sharing a preprint of \cite{bhattacharya2017fully} with us, and Virginia Vassilevska Williams for sharing a preprint of \cite{stubbs2017metatheorems} with us. The fifth author wishes to thank Seffi Naor for a remark which inspired \Cref{def:approx-max}.


\appendix
\section*{Appendix}
\section{Properties of the Nice Partition of \citet{bhattacharya2017fully}}\label{sec:justification}

In this Section we justify the use of the fully-dynamic algorithm of \cite{bhattacharya2017fully} for  maintaining a nice partition as per Definition \ref{def:nice-partition}, where the worst-case update time for inserting or deleting an edge is $O(\log^3 n)$ for a fixed constant $\epsilon > 0$.
Our goal here is twofold: first, to claim that the number of edge weight changes in each update operation is bounded by
$O(\beta^2 \log n)$ in the worst case and second, that the update time is $O(\beta^7 \log^3{n})$ in the worst case.
Although the worst-case update time in \cite{bhattacharya2017fully} is $O(\log^3 n)$ (ignoring $\beta$ factors), the number of changes of the edge weights during an update could be much larger, potentially polynomial in $n$. This can happen if, for example, changes of edge weights are maintained implicitly in an aggregated or lazy fashion. Specifically, perhaps the changes of edge weights are implicitly maintained by vertices changing their levels during an update call, so that a small change in the level of a vertex hides a polynomial number of changes of edge weights. Fortunately, this is not the case in \cite{bhattacharya2017fully}, as we assert in the following two theorems, justifying our use of the result of \cite{bhattacharya2017fully} for dynamically maintaining a nice partition.
We note that this is crucial for our needs, as the runtime of our algorithm depends on the number of edge weight changes.

First of all, it is easy to observe that the number of changes to edge weights is bounded by the worst-case runtime of an update procedure which is $O(\log^3 n)$ for a fixed constant $\epsilon > 0$. According to Section 3.4 in \cite{bhattacharya2017fully}, every edge $(x,y) \in E$ maintains the values of its weight $w(x,y)$ and level $\ell(x,y)$, and thus it is straightforward that there are at most $O(\log^3 n)$ changes of edge weights per update, for constant $\epsilon>0$. We further prove in Lemma \ref{lem:nice-partition-weight-changes} that the number of changes of edge weights per insert/delete is actually bounded by $O(\log n)$  per insert/delete for a constant $\epsilon>0$.

\begin{lem} \label{lem:nice-partition-weight-changes}
The algorithm in \cite{bhattacharya2017fully} dynamically maintains a nice partition while edges are inserted and deleted from the graph.
Let $\{ w_e \}_{e \in E}$ be the edge weights just before an update (edge insertion or deletion), and let $\{w'_e\}_{e \in E}$ be the edge weights just after the update. Then the number of changes in edge weights ({\sl i.e.}, the number of edges $e \in E$ such that $w_e \ne w'_e$) is bounded by $O(\beta^2 \log n)$.
\end{lem}

\begin{proof}
The dynamic nice partition is maintained in \cite{bhattacharya2017fully} as follows: before and after an edge insertion/deletion each vertex is a ``clean'' vertex in exactly one of six states {\textbf{UP,DOWN,UP-B,DOWN-B,SLACK and IDLE}}. Each has several constraints associated with it, and it is immediate that if the constraints of the state of every vertex hold then the resulting partition and assigned levels and edge weights form a nice-partition as per Definition \ref{def:nice-partition} (see \cite[Lemma 3.2]{bhattacharya2017fully}).

Each insertion (or deletion) of an edge to (from) the graph increases (resp. decreases) the weight associated with each vertex touching the edge. The insertion/deletion of an edge may cause an endpoint $x$ of the edge to be marked as {\sl dirty} based on a set of rules (rules 3.1-3.3 in \cite{bhattacharya2017fully}). Intuitively, a dirty vertex requires some fixing; for example, if a vertex $x$ has a large weight $W_x$ very close to $1$ then an insertion of an edge touching $x$ further increases the weight $W_x$ making it even further closer to $1$ (or even larger than $1$). To preserve the constraints of all the states, this requires fixing the dirty vertex, e.g., by increasing the shadow level of $x$ with respect to some edge $(x,y)$ touching it, which reduces the weight the edge $(x,y)$ and thus reduces the weight $W_x$ of the vertex $x$. As a side effect of reducing the weight of $(x,y)$, also the weight $W_y$ of the vertex $y$ reduces, possibly causing $y$ to become a dirty vertex which requires some fixing.

The above description forms a chain of \emph{activations} (a vertex $x$ is said to be \emph{activated} if its weight $W_x$ changes). Assume that the inserted/deleted edge touches a vertex $x_1$ and thus $x_1$ becomes activated as it has its weight $W_{x_1}$ changed. It may happen (according to rules 3.1-3.3 in \cite{bhattacharya2017fully}) that due to the change in $W_{x_1}$ the vertex $x_1$ becomes dirty and this requires fixing $x_1$ by calling the procedure \textbf{FIX-DIRTY-NODE($x_1$)}. The procedure \textbf{FIX-DIRTY-NODE($x_1$)} fixes the vertex $x_1$ so it becomes clean again and by that guarantees that it satisfies its state's constraints.

During \textbf{FIX-DIRTY-NODE($x_1$)} either none of the edges changed their weights (and thus the chain of activations terminates), or exactly one edge $(x_1,x_2)$ has its weight changed. In the latter case, we say that $x_2$ incurs an induced activation, that may cause $x_2$ to become dirty as per rules 3.1-3.3 in \cite{bhattacharya2017fully}. In this case, the chain of activations continues as we now call \textbf{FIX-DIRTY-NODE($x_2$)}, which in turn may cause at most one edge $(x_2, x_3)$ to change its weight, and then vertex $x_3$ occurs an induced activation. The chain of activations continues, and as long as there is a dirty vertex $x_i$ the procedure \textbf{FIX-DIRTY-NODE($x_i$)} is called, restoring $x_i$ to clean status but potentially changing the weight of at most one edge $(x_i,x_{i+1})$, which may cause $x_{i+1}$ to become a dirty vertex, an so on. However, as stated in the conference version in \cite[Theorem 6.2]{bhattacharya2017fully} and proved in the full version (see \cite[Theorem 6.4 and Lemma 8.1]{bhattacharya2017fullyfully}), the chain of activations comprises at most $O(\log_\beta n) = O(\log n)$ activations; i.e., there are at most $O(\log n)$ calls to \textbf{FIX-DIRTY-NODE($x$)} in the chain until all vertices are clean. Furthermore, according to Assumptions 1 and 2 in  \cite[Section 6]{bhattacharya2017fully}, an insertion or deletion of an edge may cause at most $O(\beta^2)$ chains of activations, and thus following an edge insertion or deletion there are at most $O(\beta^2 \log n)$ calls to \textbf{FIX-DIRTY-NODE($x$)}.

Observe that \textbf{FIX-DIRTY-NODE} is the only procedure that may cause a change in the weight of an edge (except for the insertion/deletion of the edge itself), and each call to \textbf{FIX-DIRTY-NODE} changes the weight of at most one edge. Thus, we conclude that the number of changes of edge weights during an insertion/deletion of an edge is at most the number of calls to \textbf{FIX-DIRTY-NODE}, and thus the number of changes of edge weights is bounded by $O(\beta^2\log n)$.  We mention that the level of a node may change due to the call of the subroutine \textbf{UPDATE-STATUS($x$)} (e.g., \textbf{Cases 2-a and 2-b}), but this does not cause a change in the weights of edges incident on $x$.
\end{proof}

Finally, we furthermore state the worst-case update time of an update for maintaining the nice partition. While \cite{bhattacharya2017fully} proves that the worst-case update time is $O(\log^3 n)$ for constant $\beta \ge 5$, we also mention the dependency of the update time in $\beta$ which is implicit in \cite{bhattacharya2017fully}.

\begin{lem}\label{lem:wc-updates}
The algorithm in \cite{bhattacharya2017fully} dynamically maintains a nice partition in  $O(\beta^{7}\log^3 n )$ worst-case update time.
\end{lem}

\begin{proof}
The worst-case update time is bounded by the number of calls to \textbf{FIX-DIRTY-NODE($x$)} times the runtime of  \textbf{FIX-DIRTY-NODE($x$)}. As mentioned in the proof of Theorem \ref{lem:nice-partition-weight-changes}, each edge insertion/deletion may invoke $O(\beta^2 \log n)$ calls to \textbf{FIX-DIRTY-NODE($x$)}, and the runtime of each \textbf{FIX-DIRTY-NODE($x$)} is bounded by $O(\beta^5 \log^2 n)$, since the runtime of \textbf{FIX-DIRTY-NODE($x$)} in the worst-case is dominated by the runtime of the procedure \textbf{FIX-DOWN($x$)} which takes $O(\beta^5 \log^2 n)$ according to \cite[Lemma 5.1 and Lemma 5.4]{bhattacharya2017fully}. Thus, the worst-cast update time of an insertion/deletion of an edge is $O(\beta^7 \log^3 n)$.
\end{proof}

Combining \Cref{lem:wc-updates} and \Cref{lem:nice-partition-weight-changes} we immediately obtain \Cref{lem:two-polylog}, restated below.

\fractwopolylog*

\bibliographystyle{acmsmall}
\bibliography{dynamic-matching,matching,streaming-matching,shiri}

\end{document}